\documentclass[11pt]{article}

\usepackage{amsmath,amssymb,amsthm,mathrsfs, ulem}
\usepackage{enumerate, hyperref}
\usepackage{xcolor}

\newcommand{\E}{\mathbb E}

\newtheorem{theorem}{Theorem}
\newtheorem{example}{Example}
\newtheorem{lemma}{Lemma}

\newtheorem{proposition}{Proposition}
\newtheorem{assumption}{Assumption}
 \newtheorem{remark}{Remark}
\newcommand{\tr}{\textcolor{black}}

\title{ High-Order Expansions of the Optimizer Map via Bell Polynomials}

\author{
Oleksii Mostovyi\thanks{%
University of Connecticut, Department of Mathematics,
Storrs, CT 06269, United States.
Email: \texttt{oleksii.mostovyi@uconn.edu}.
This research has been partially supported by the National Science Foundation
under grant No.~DMS-1848339.}
\and
Thaleia Zariphopoulou\thanks{%
Departments of Mathematics and IROM, McCombs School of Business,
The University of Texas at Austin, and the Oxford-Man Institute of
Quantitative Finance, University of Oxford.
Email: \texttt{zariphop@math.utexas.edu}.}
}

\date{}
\begin{document}

\maketitle
\begin{abstract}
 Completely monotonic inverse marginal (CMIM) utilities, introduced in \cite{MSZ}, constitute a tractable class of preferences that includes many of the most important utility functions used in mathematical finance, such as power and exponential utilities. In stochastically dominant markets, their Bernstein representation induces a hidden linear structure in the dual optimization problem that serves as the foundation for the present analysis.

In this paper, we investigate the sensitivity of optimal investment with respect to perturbations of investor preferences within the CMIM class. Exploiting Bernstein's representation theorem, we show that, under stochastic dominance, affine perturbations of Bernstein measures induce an affine representation of the dual value function. As a result, the dependence of the optimization problem on preferences can be analyzed through a scalar budget equation, allowing us to prove analyticity of the associated Lagrange multiplier with respect to the perturbation parameter and to derive convergent analytic expansions of arbitrary order for the primal value function and the optimal terminal wealth, with explicit recursive formulas expressed through Bell polynomials.


\end{abstract}
\section{Introduction}


The choice of utility function plays a fundamental role in continuous-time portfolio selection. In general semimartingale models, however, the dependence of optimal investment strategies and value functions on investor preferences is highly nonlinear. This raises the question of identifying classes of utility functions that are sufficiently rich for applications while retaining enough structure to permit sensitivity analysis beyond the first order.

One such class is that of completely monotonic inverse marginal (CMIM) utilities, introduced in \cite{MSZ}. This class contains many of the most important utility functions used in mathematical finance, including logarithmic, power, and exponential utilities, and enjoys a representation through Bernstein measures. In stochastically dominant markets, this representation leads to a hidden linear structure of the dual optimization problem and, consequently, to a degree of tractability that appears to be unavailable for general utility functions.

The principal insight of \cite{MSZ} is that, in stochastically dominant markets, the Bernstein representation of CMIM utilities induces a linear structure in the dual problem. Whereas \cite{MSZ} exploited this structure to study analyticity with respect to initial wealth, the present paper investigates perturbations of investor preferences induced by affine perturbations of Bernstein measures. 
Combining stochastic dominance with the Bernstein representation of CMIM utilities, we establish an affine representation of the dual value function and exploit this structure to derive convergent recursive expansions of arbitrary order for the Lagrange multiplier, the primal value function, and the optimal terminal wealth. By contrast, even obtaining first- and second-order expansions in general incomplete markets is highly nontrivial; see \cite{KS06, MostovyiSirbuPertSem}.

%
%
%

The appearance of Bell polynomials reflects an intrinsic compositional structure of the optimization problem. The Lagrange multiplier is determined implicitly through the budget equation, while the optimizer is obtained by evaluating the inverse marginal utility at the optimally scaled distinguished dual element. As a result, higher-order sensitivity coefficients arise from repeated differentiation of compositions, leading naturally to the Fa\`a di Bruno formula and Bell polynomials. To the best of our knowledge, Bell polynomials have not previously been used systematically in portfolio optimization.

Several questions remain open. These include approximation of general preferences by CMIM utilities, statistical elicitation and calibration of Bernstein measures, and higher-order expansions under simultaneous perturbations of probabilities and preferences. We leave these topics for future work.

The remainder of the paper is organized as follows. Section \ref{secModel} introduces the utility maximization framework, the class of CMIM utilities, and affine perturbations of Bernstein measures. Section \ref{secMain} establishes the affine structure of the dual value function and develops convergent analytic expansions of arbitrary order for the Lagrange multiplier, the primal value function, and the optimal terminal wealth. Section \ref{secCRRAex} illustrates the theory in the case of affine perturbations of CRRA inverse marginal utilities and derives explicit formulas in terms of moments of the distinguished dual optimizer. Section \ref{exBS} studies simultaneous perturbations of market beliefs and investor preferences and obtains first-order sensitivity formulas in a Black-Scholes-Merton setting. The appendix \ref{sec:Appendix} recalls the Bell polynomials and Fa\`a di Bruno formula used throughout the paper.

\section{Setup}\label{secModel}
\subsection{Semimartingale market and dual domain}
 
Let \(S=(S^1,\dots,S^d)\) be a \(d\)-dimensional semimartingale on a
filtered probability space
$(\Omega,\mathcal F,(\mathcal F_t)_{0\leq t\leq T},\mathbb P)$,
satisfying the usual conditions. We interpret \(S\) as the price process of the risky assets discounted by the units of the riskless asset.
For \(x>0\), let \(\mathcal X(x)\) denote the set of nonnegative wealth
processes of the form
\[
X_t=x+\int
_
0^t H_u\,dS_u,\qquad 0\le t\le T,
\]
where \(H\) is predictable and \(S\)-integrable.
The dual domain is
\[
\mathcal Y
:=
\left\{
Y\ge 0:\ Y_0=1,\ XY \text{ is a supermartingale for every }
X\in\mathcal X(1)
\right\}.
\]
 
The utility-maximization problem in semimartingale models has been studied
extensively in the literature. We refer to
 \cite{KS99, KZ03, KS2003,  Zit05, KS06, Mostovyi2015, BiaginiCerny2020, Mon22}
for general background on primal-dual formulations, existence of optimizers,
and sensitivity analysis in incomplete markets. 

\subsection{Stochastically dominant markets}
Throughout the paper, we
work within the framework of stochastically dominant markets, 
that is, the ones satisfying the following assumption. We note that stochastically dominant markets in the context of analyticity of the value function have been studied in  \cite{MSZ}.

\begin{assumption}[Stochastic dominance of the market]
\label{asStDom}
The dual domain \(\mathcal Y\) is nonempty, and the market is
stochastically dominant in the sense of \cite{MSZ}. That is, there
exists \(\widehat Y\in\mathcal Y\) such that
\[
\int
_
0^\infty \mathbb P[\widehat Y_T>x]\,dx
\ge
\int
_
0^\infty \mathbb P[Y_T>x]\,dx,
\qquad Y\in\mathcal Y.
\]
\end{assumption}
We note that non-emptiness of $\mathcal Y$ is required for the absence of arbitrage-type condition  NUPBR  introduced in \cite{KaratzasKardaras2007}, which in turn is crucial for the non-degeneracy of the utility maximization problem, see \cite[Proposition 4.19]{KaratzasKardaras2007}. 
\begin{remark} By \cite[Proposition 7.2]{MSZ}, Assumption \ref{asStDom} is equivalent to the
seemingly weaker notion of infinite-order stochastic dominance. We use
the present formulation because it is more convenient for the subsequent
analysis.
\end{remark}

\subsection{The CMIM class}
 
The class of completely monotonic inverse marginal utilities (CMIM) and its
connection with stochastically dominant markets were introduced in
\cite{MSZ}. One of the key features of this framework is that the dual
optimizer is independent of the particular utility function within the CMIM
class, leading to strong regularity properties of the associated value
functions.

\subsection{Bernstein representation}

The tractability of optimization problems with CMIM utilities is closely connected to the Bernstein theorem (see, e.g., \cite{Widder1946, SSV2012, Merkle2014}), which asserts that every completely monotonic function admits a unique representation as the Laplace transform
of a Radon measure on $[0,\infty)$. 
For every \(I\in\mathcal C\), Bernstein's theorem yields the representation
\begin{equation}\label{defI}
I(y)=\int_0^\infty e^{-yz}\,\mu(dz),
\qquad y>0,
\end{equation}
for a nonnegative sigma-finite measure $\mu$ on $[0, \infty)$ such that the integral
converges for every $x > 0$.

\subsection{Affine perturbations of Bernstein measures}
Let
\[
\mathcal C
:=
\left\{
I:(0,\infty)\to(0,\infty):
I \text{ is completely monotonic},
\ \lim_{y\downarrow 0} I(y)=\infty,
\ \lim_{y\uparrow\infty} I(y)=0
\right\}.
\]
Fix \(\varepsilon_0>0\) and consider a family of utility functions
\[
U^\varepsilon:(0,\infty)\to\mathbb R,
\qquad \varepsilon\in(-\varepsilon_0,\varepsilon_0),
\]
which are strictly concave, strictly increasing, continuously
differentiable, satisfy the Inada conditions, and whose inverse marginal
utilities
\begin{equation}\label{Ie}
I^\varepsilon:=(U^{\varepsilon\prime})^{-1} ~~\text{ belong~to~} \mathcal C~\text{ for~every~}\varepsilon\in(-\varepsilon_0,\varepsilon_0).
\end{equation}
The following examples, adapted from \cite{MSZ}, illustrate the Bernstein representation for several classical utility functions. They also motivate the perturbation framework introduced in Assumption \ref{asBernstein}.
\begin{example}\label{exPower}
\begin{enumerate}[(i)]\item
Let $\mu(dz) = dz$, then \eqref{defI} yields  $$I(y) = \frac 1y, \qquad y>0,$$ 
which is the inverse marginal utility of logarithmic utility. 
\item For $p<1, p\neq 0$, let $q$ be given via
$$\frac 1q + \frac 1p = 1,$$
i.e., $q = -\frac{p}{1 - p}$. Then, for 
$\mu(dz) = \frac 1{\Gamma(1-q)}z^{-q}dz$, where  $\Gamma$ is the Gamma function, see \cite[p. viii, formula (2)]{SSV2012}, 
\eqref{defI} gives $$I(y) = y^{\frac 1{p-1}},\qquad y>0,$$
corresponding to power utility. 
\item Let $p_i<1, p_i\neq 0$, let $q_i = -\frac{p_i}{1 - p_i}$, and for some $N\in\mathbb N$, let 
$$\mu = \sum\limits_{i = 1}^{N} \mu_i,\qquad 
 where \qquad \mu_i(dz) = \frac 1{\Gamma(1-q_i)}z^{-q_i}dz.$$ In this case,
$$I(y) = \sum\limits_{i = 1}^{N} y^{\frac 1{p_i-1}},\qquad y>0,$$
and so $I\in\mathcal C$.  
\end{enumerate}
\end{example}


These examples will be revisited in Section \ref{secCRRAex}, where the general results yield explicit expansions. 
 For the sensitivity analysis, we consider affine perturbations of Bernstein measures, thus \eqref{Ie} will be sharpened via   the following assumption. 
\begin{assumption}[Affine perturbations of Bernstein measures]
\label{asBernstein}
There exist 
 a Radon measure \(\mu^0\) on
\((0,\infty)\), and a signed Radon measure \(\nu\) on \((0,\infty)\) such that,
for every \(\varepsilon\in(-\varepsilon_0,\varepsilon_0)\),
\begin{equation} \label{pertParam}
\mu^\varepsilon=\mu^0+\varepsilon\nu
\end{equation}
is a nonnegative Radon measure on \((0,\infty)\), and
\[
I^\varepsilon(y)
=
\int_0^\infty e^{-yz}\,\mu^\varepsilon(dz),
\qquad y>0.
\]

Moreover,
\[
\int_0^\infty e^{-yz}\,|\nu|(dz)<\infty,
\qquad y>0,
\]
and there exists \(y_0>0\) such that
\begin{equation}\label{integrabilityOfnu}
\int_0^\infty \frac{e^{-zy_0}}{  z}\,|\nu|(dz)<\infty,
\qquad
\int_0^\infty \frac{e^{-zy_0}}{  z}\,\mu^0(dz)<\infty.
\end{equation}
 \end{assumption}
 
We denote by \(\mathcal M\) the class of nonnegative Radon measures \(\mu\) on \((0,\infty)\) whose Laplace transforms belong to \(\mathcal C\) and such that, for some \(y_0>0\), \[ \int_0^\infty \frac{e^{-zy_0}}{z}\,\mu(dz)<\infty.\]
 
 \begin{remark} As pointed out in Assumption \ref{asBernstein}, 
the representing Bernstein measure is
supported on $(0,\infty)$, since an atom at $0$ would imply
$$\lim_{y\to\infty}I(y)>0,$$
contradicting the Inada condition.  
\end{remark}

\subsection{Primal and dual value functions}


For \((y,\varepsilon)\in(0,\infty)\times(-\varepsilon_0,\varepsilon_0)\), define the
convex conjugate
\[
V^\varepsilon(y)
:=
\sup_{x>0}\bigl(U^\varepsilon(x)-xy\bigr).
\]

The primal and dual value functions are
\begin{equation}\label{defv}
u(x,\varepsilon):=
\sup_{X\in\mathcal X(x)}
\mathbb E\!\left[U^\varepsilon(X_T)\right],
\qquad
(x,\varepsilon)\in(0,\infty)\times(-\varepsilon_0,\varepsilon_0),
\end{equation}
and
\begin{equation}\label{defv}
v(y,\varepsilon)
:=
\inf_{Y\in\mathcal Y}
\mathbb E\!\left[V^\varepsilon(yY_T)\right],
\qquad
(y,\varepsilon)\in(0,\infty)\times(-\varepsilon_0,\varepsilon_0).
\end{equation}

Since  perturbations are parametrized on the level of the derivative of the utility functions, the latter are specified up to an additive constant.
We fix this ambiguity by choosing a reference point
\(y_*>0\) and imposing
\begin{equation}\label{Vnorm}
V^\varepsilon(y_*)=0, \qquad \varepsilon\in(-\varepsilon_0,\varepsilon_0).
\end{equation}
Equivalently,
\[
V^\varepsilon(y)
=
\int_0^\infty
\frac{e^{-yz}-e^{-y_*z}}{z}
\,\mu^\varepsilon(dz),
\qquad y>0.
\]

\begin{remark}
The normalization \eqref{Vnorm} is imposed without loss of generality.
Indeed, replacing a utility function by a positive affine transformation
\[
\widetilde U(x)=aU(x)+c,\qquad a>0,
\]
does not affect the optimizer or the optimal trading strategy. It merely
rescales and shifts the value function,
\[
\widetilde u(x,\varepsilon)=au(x,\varepsilon)+c.
\]
The corresponding convex conjugate satisfies
\[
\widetilde V(y)=aV\!\left(\frac{y}{a}\right)+c,
\]
so the additive normalization of the dual function is arbitrary. 
Condition \eqref{Vnorm} simply fixes the additive normalization of the dual function.
\end{remark}
\begin{assumption}[Dual finiteness]
\label{finValue}
Assume that
\[
v(y,\varepsilon)<\infty,
\qquad
(y,\varepsilon)\in (0,\infty)\times(-\varepsilon_0,\varepsilon_0).
\]
\end{assumption}
 
%
%

Under Assumptions~\ref{asStDom} and \ref{finValue}, e.g., following the argument in \cite{MSZ}, one can show that the same element
\(\widehat Y\in\mathcal Y\) attains the dual infimum for every preference
specification under consideration. Hence
\[
v(y,\varepsilon)
=
\mathbb E\!\left[V^\varepsilon(y\widehat Y_T)\right],
\qquad
(y,\varepsilon)\in (0,\infty)\times(-\varepsilon_0,\varepsilon_0).
\]

\subsection{Budget equation and affine dual linearity}

For \((x,\varepsilon)\in(0,\infty)\times(-\varepsilon_0,\varepsilon_0)\), the
Lagrange multiplier \(y(x,\varepsilon)\) is determined by the budget
equation
\begin{equation}\label{budeq}
\mathbb E\left[
\widehat Y_T I^\varepsilon
\bigl(y(x,\varepsilon)\widehat Y_T\bigr)
\right]=x.
\end{equation}
Using the Bernstein representation of \(I^\varepsilon\) and Tonelli's theorem, the budget equation \eqref{budeq} can be rewritten as
\[
\int_0^\infty
\mathbb E\left[
\widehat Y_T e^{-y(x,\varepsilon)\widehat Y_Tz}
\right]\mu^\varepsilon(dz)=x.
\]

%
%
%
%
%
%
%
%
%

\section{Affine Structure and Analyticity}\label{secMain}
 Let  
\begin{equation}\label{defg0}
g_0(y):=\E\!\left[\widehat Y_T I^0(y\widehat Y_T)\right]=
\int_0^\infty
\E\left[
\widehat Y_T e^{-y\widehat Y_T z}
\right]\mu^0(dz),
\end{equation}
\begin{equation}\label{defJ}
J(y):=\int_0^\infty e^{-yz}\nu(dz),
\end{equation}
 and
 \begin{equation}\label{defg1}
g_1(y):=\E\!\left[\widehat Y_T J(y\widehat Y_T)\right]=
\int_0^\infty
\E\left[
\widehat Y_T e^{-y\widehat Y_T z}
\right]\nu(dz).
\end{equation}

\begin{theorem}[Affine structure of the dual value]\label{thm1}
Assume that Assumptions~\ref{asStDom}, \ref{asBernstein}, and \ref{finValue}, and the normalization condition \eqref{Vnorm} hold. 
Then, for every \(y>0\), the dual value function admits the affine representation 
\[
v(y,\varepsilon)=v_0(y)+\varepsilon v_1(y),
\]
where
\[
v_0(y):
=
\E\left[
\int_0^\infty
\frac{e^{-y\widehat Y_Tz}-e^{-y_*z}}{z}\,\mu^0(dz)
\right],
\]
and
\[
v_1(y):
=
\E\left[
\int_0^\infty
\frac{e^{-y\widehat Y_Tz}-e^{-y_*z}}{z}\,\nu(dz)
\right].
\]
Consequently, we have
\[
-v_y(y,\varepsilon)=g_0(y)+\varepsilon g_1(y).
\]

\end{theorem}

\begin{proof}
By Assumptions~\ref{asStDom} and~\ref{finValue},
\[
v(y,\varepsilon)
=
\E\left[
V^\varepsilon(y\widehat Y_T)
\right].
\]
Using the normalization condition \eqref{Vnorm}, which asserts that \(V^\varepsilon(y_*)=0\), we have
\[
V^\varepsilon(y)
=
\int_0^\infty
\frac{e^{-yz}-e^{-y_*z}}{z}\,\mu^\varepsilon(dz).
\]
Since 
\[
\mu^\varepsilon=\mu^0+\varepsilon\nu,
\]
it follows that
\[
v(y,\varepsilon)
=
v_0(y)+\varepsilon v_1(y).
\]

Differentiating with respect to \(y\), and using
\[
(V^\varepsilon)'(y)=-I^\varepsilon(y),
\]
we obtain
\[
-v_y(y,\varepsilon)
=
\E\left[
\widehat Y_T I^\varepsilon(y\widehat Y_T)
\right].
\]
By Assumption~\ref{asBernstein}, we have 
\[
I^\varepsilon=I^0+\varepsilon J,
\]
where $J$ is defined in \eqref{defJ}.
Therefore,
\[
-v_y(y,\varepsilon)
=
\E\left[
\widehat Y_T I^0(y\widehat Y_T)
\right]
+
\varepsilon
\E\left[
\widehat Y_T J(y\widehat Y_T)
\right].
\]
The integral representations of \(g_0\) and \(g_1\) follow from the
Bernstein representations of \(I^0\) and \(J\).
\end{proof}


%

The next theorem establishes the analyticity of the Lagrange multiplier. For notational simplicity, we suppress the dependence of the expansion coefficients on the initial wealth $x$.

\begin{theorem}\label{thm2}(Analyticity and recursive expansion of the Lagrange multiplier)
Under the conditions of Theorem \ref{thm1}, fix \(x>0\), and let \(y_0=y(x, 0) = u_x(x, 0)\) be the unique solution of
\[
g_0(y_0)=x,
\]
where $g_0$ is given by \eqref{defg0}.
Then
\[
g_0'(y_0)<0.
\]
Consequently, with $g_1$ defined in \eqref{defg1}, the equation
\[
g_0(y(x,\varepsilon))
+
\varepsilon g_1(y(x,\varepsilon))
=
x
\]
has a unique analytic solution near \(\varepsilon=0\). Namely,
\[
y(x,\varepsilon)=
\sum_{m=0}^\infty
\frac{\varepsilon^m}{m!}y_m .
\]
The coefficients are determined recursively as follows. For \(m\ge1\),
\begin{equation}\label{defym}
y_m
=
-\frac{A_m}{g_0'(y_0)},
\end{equation}
where
\[
A_m
=
\sum_{r=2}^{m}
g_0^{(r)}(y_0)
B_{m,r}(y_1,\ldots,y_{m-r+1})
+
m\sum_{r=1}^{m-1}
g_1^{(r)}(y_0)
B_{m-1,r}(y_1,\ldots,y_{m-r}).
\]
Here \(B_{n,k}\) denotes the partial Bell polynomial, as presented in Appendix \ref{sec:Appendix}, with
\[
B_{0,0}=1,
\qquad
B_{n,0}=0,\quad n\ge1.
\]
\end{theorem}

\begin{proof}

By Theorem \ref{thm1}, we have 
\[
g_0(y)=\E[\widehat Y_T I^0(y\widehat Y_T)].
\]
Since \(I^0\in\mathcal C\), Bernstein's theorem implies
\[
(I^0)'(u)<0,\qquad u>0.
\]
Therefore
\[
g_0'(y)
=
\E\left[\widehat Y_T^2(I^0)'(y\widehat Y_T)\right]
<0,
\]
and in particular
\[
g_0'(y_0)<0.
\]
Moreover, by their integral representations in \eqref{defg0} and \eqref{defg1}, $g_0$ and $g_1$
 are analytic on $(0,\infty)$, as follows from the standard analyticity properties of Laplace transforms.

Define
\[
G(y,\varepsilon)
=
g_0(y)+\varepsilon g_1(y)-x.
\]
Then
\[
G(y_0,0)=0,
\qquad
G_y(y_0,0)=g_0'(y_0)<0.
\]
By the analytic implicit function theorem, there exists a unique analytic
function
\[
\varepsilon\mapsto y(x,\varepsilon)\]
near \(0\) such that
\[
G(y(x,\varepsilon), \varepsilon)=0.
\]

It remains to derive the recursive formula for the coefficients. Write
\[
y(x,\varepsilon)=
\sum_{m=0}^{\infty}\frac{\varepsilon^m}{m!}y_m .
\]
Differentiating
\[
g_0(y(x,\varepsilon))
+
\varepsilon g_1(y(x,\varepsilon))
-
x
=
0
\]
\(m\) times with respect to \(\varepsilon\), evaluating at
\(\varepsilon=0\), and applying the Faà di Bruno formula gives
\[
g_0'(y_0)y_m+A_m=0,
\]
where
\[
A_m
=
\sum_{r=2}^{m}
g_0^{(r)}(y_0)
B_{m,r}(y_1,\ldots,y_{m-r+1})
+
m\sum_{r=1}^{m-1}
g_1^{(r)}(y_0)
B_{m-1,r}(y_1,\ldots,y_{m-r}).
\]
Since \(g_0'(y_0)<0\), this yields
\[
y_m=-\frac{A_m}{g_0'(y_0)}.
\]
\end{proof}
 
\begin{remark}
The strict negativity of \(g_0'\) established in the proof of Theorem~\ref{thm2}
provides an alternative to the second-order differentiability theory
developed by Kramkov and S\^{\i}rbu \cite{KS06}.

Indeed, if the utility function \(U^0\) had bounded relative risk aversion,
that is,
\[
0<c
\le
-\frac{xU^{0\prime\prime}(x)}{U^{0\prime}(x)}
\le C
<\infty,
\qquad x>0,
\]
then the results of \cite{KS06} would imply the twice differentiability of
the dual value function \(v_0\), and in particular the positivity of
\[
v_0''(y)>0,
\qquad y>0.
\]
Since
\[
g_0=-v_0',
\]
this would immediately yield
\[
g_0'(y)=-v_0''(y)<0.
\]

However, the CMIM condition, in general, does not imply bounded relative risk aversion.
Indeed, if
\[
I^0(y)=\int_0^\infty e^{-yz}\,\mu^0(dz),
\]
then
\[
-\frac{xU^{0\prime\prime}(x)}{U^{0\prime}(x)}
=
-\frac{I^0(y)}{y(I^0)'(y)}
=
\frac{
\displaystyle \int_0^\infty e^{-yz}\,\mu^0(dz)
}{
\displaystyle y\int_0^\infty z e^{-yz}\,\mu^0(dz)
},
\qquad x=I^0(y),
\]
and the latter quantity need not be bounded above or below away from zero.

Consequently, the argument used in the proof of Theorem~\ref{thm2} relies directly
on complete monotonicity and Bernstein's representation theorem rather than
on bounded-risk-aversion assumptions.
\end{remark}

In particular, the first coefficients are 
\[y_0=y(x, 0),\]
\[
y_1
=
-\frac{g_1(y_0)}{g_0'(y_0)}.
\]

\[
y_2
=
-\frac{
g_0''(y_0)y_1^2
+
2g_1'(y_0)y_1
}{
g_0'(y_0)
}.
\]

\[
y_3
=
-\frac{
g_0'''(y_0)y_1^3
+
3g_0''(y_0)y_1y_2
+
3g_1''(y_0)y_1^2
+
3g_1'(y_0)y_2
}{
g_0'(y_0)
}.
\]

\begin{remark}

The classical identities
\[
u(x,\varepsilon)=
v(y(x,\varepsilon), \varepsilon)
+
xy(x,\varepsilon),
\]
where \(y(x,\varepsilon) = u_x(x,\varepsilon)\), 
and
\[
\widehat X_T(x,\varepsilon)=
I^\varepsilon(y(x,\varepsilon)\widehat Y_T)
\]
indicate that the analyticity of the Lagrange multiplier should propagate
to the primal value function and the optimizer. Theorems~\ref{thm3} and~\ref{thm4} make
this precise.
\end{remark}
We start from an auxiliary lemma.

\begin{lemma}\label{lemAPC} (Affine perturbations and analytic composition)
Let
\[
 F^\varepsilon(z)=F_0(z)+\varepsilon F_1(z), \qquad z(\varepsilon)
=
\sum_{m=0}^\infty
\frac{\varepsilon^m}{m!}z_m, \qquad \varepsilon\in(-\varepsilon_0, \varepsilon_0), 
\]
where \(F_0\) and \(F_1\) are analytic functions, and $\varepsilon\to z(\varepsilon)$
is analytic in a neighborhood of \(\varepsilon=0\).
Then
\[
F^\varepsilon(z(\varepsilon))
=
\sum_{m=0}^\infty
\frac{\varepsilon^m}{m!}F_m,
\]
where
\[
F_0=F_0(z_0),
\]
and, for \(m\geq 1\),
\[
F_m
=
\sum_{r=1}^m
F_0^{(r)}(z_0)
B_{m,r}
(z_1,\dots,z_{m-r+1})
+
m
\sum_{r=0}^{m-1}
F_1^{(r)}(z_0)
B_{m-1,r}
(z_1,\dots,z_{m-r}),
\]
where \(B_{n,k}\) denotes the partial Bell polynomial.

\end{lemma}

\begin{proof}
By the Fa\`a di Bruno formula,
\[
F_0(z(\varepsilon))
=
\sum_{m=0}^\infty
\frac{\varepsilon^m}{m!}
\sum_{r=1}^m
F_0^{(r)}(z_0)
B_{m,r}
(z_1,\dots,z_{m-r+1}),
\]
while
\[
\varepsilon F_1(z(\varepsilon))
=
\sum_{m=1}^\infty
\frac{\varepsilon^m}{m!}
m
\sum_{r=0}^{m-1}
F_1^{(r)}(z_0)
B_{m-1,r}
(z_1,\dots,z_{m-r}).
\]
Adding the two expansions yields the result.
\end{proof}

The preceding lemma explains the similarity between the coefficient formulas
for the primal value function
\[
u(x,\varepsilon)=
v_0(y(x,\varepsilon))
+
\varepsilon v_1(y(x,\varepsilon))
+
xy(x,\varepsilon),
\]where \(y(x,\varepsilon) = u_x(x,\varepsilon)\), 
and those for the optimal terminal wealth
\[
\widehat X_T(x,\varepsilon)=
I^0(y(x,\varepsilon)\widehat Y_T)
+
\varepsilon
J(y(x,\varepsilon)\widehat Y_T).
\]

Indeed, both are particular cases of the composition
\[
F^\varepsilon(z(\varepsilon))
=
F_0(z(\varepsilon))
+
\varepsilon F_1(z(\varepsilon)),
\]
with
$
(F_0,F_1)
=
(v_0,v_1)
$
in the first case and
$
(F_0,F_1)
=
(I^0,J)
$
in the second.

The next theorem establishes the analyticity of the primal quantities.
\begin{theorem}\label{thm3}(Analyticity of the primal value function)
Assume the conditions of Theorem~2. Then, for every \(x>0\),
\[
u(x,\varepsilon)=
v(y(x,\varepsilon), \varepsilon)
+
xy(x,\varepsilon)\]where \(y(x,\varepsilon) = u_x(x,\varepsilon)\), 
is analytic in \(\varepsilon\) near \(0\). Moreover,
\[
u(x,\varepsilon)=
\sum_{m=0}^{\infty}
\frac{\varepsilon^m}{m!}u_m(x).
\]

Let
\[
y(x,\varepsilon)=
\sum_{m=0}^{\infty}
\frac{\varepsilon^m}{m!}y_m,
\]
where the coefficients \(y_m\) are given by Theorem~\ref{thm2}. Then

\[
u_0(x)=v_0(y_0)+xy_0.
\]
For \(m\geq1\),
\[
\begin{aligned}
u_m(x)
={}&
\sum_{r=1}^{m}
v_0^{(r)}(y_0)
B_{m,r}(y_1,\ldots,y_{m-r+1})
\\
&+
m
\sum_{r=0}^{m-1}
v_1^{(r)}(y_0)
B_{m-1,r}(y_1,\ldots,y_{m-r})
+
xy_m.
\end{aligned}
\]
Here \(B_{n,k}\) denotes the partial Bell polynomial, with
\[
B_{0,0}=1,
\qquad
B_{n,0}=0
\quad
\text{for } n\geq1.
\]

Equivalently, using the first-order condition
\[
g_0(y_0)=x,
\]
or, equivalently,
\[
v_0'(y_0)+x=0,
\]
the term involving \(y_m\) cancels. Hence, for \(m\geq1\),
\[
\begin{aligned}
u_m(x)
={}&
\sum_{r=2}^{m}
v_0^{(r)}(y_0)
B_{m,r}(y_1,\ldots,y_{m-r+1})
\\
&+
m
\sum_{r=0}^{m-1}
v_1^{(r)}(y_0)
B_{m-1,r}(y_1,\ldots,y_{m-r}).
\end{aligned}
\]
\end{theorem}

\begin{proof}
By Theorem~\ref{thm2}, the map
\[
\varepsilon\mapsto y(x,\varepsilon)\]
is analytic in a neighborhood of \(0\). 
By the normalization condition \eqref{Vnorm}, we have 
\[
v(y,\varepsilon)=v_0(y)+\varepsilon v_1(y).
\]
Therefore
\[
u(x,\varepsilon)=
v_0(y(x,\varepsilon))
+
\varepsilon v_1(y(x,\varepsilon))
+
xy(x,\varepsilon).
\]
Since \(y(x,\varepsilon)\) is analytic by Theorem~\ref{thm2}, the analyticity of
\(u(x,\varepsilon)\) follows. Applying Lemma~\ref{lemAPC} with
\[
(F_0,F_1)=(v_0,v_1),
\qquad
z(\varepsilon)=y(x,\varepsilon),
\]
gives the coefficient formula.
\end{proof}
The next theorem provides an expansion of the optimal terminal wealth. 
\begin{theorem}\label{thm4}(Analyticity of the optimizer map)
Assume the conditions of Theorem~2. Then, for every \(x>0\),
\[
\widehat X_T(x,\varepsilon)=
I^\varepsilon(y(x,\varepsilon)\widehat Y_T)
\]
is analytic in \(\varepsilon\) near \(0\). Moreover,
\[
\widehat X_T(x,\varepsilon)=
\sum_{m=0}^\infty
\frac{\varepsilon^m}{m!}X_m.
\]
Let
\[
Z_m=y_m\widehat Y_T,\qquad m\ge0,
\]
where the coefficients \(y_m\) are defined in Theorem~\ref{thm2} by \eqref{defym}. Then
\[
X_0=I^0(Z_0),
\]
and, for \(m\ge1\),
\[
X_m
=
\sum_{r=1}^m
(I^0)^{(r)}(Z_0)
B_{m,r}(Z_1,\ldots,Z_{m-r+1})
+
m\sum_{r=0}^{m-1}
J^{(r)}(Z_0)
B_{m-1,r}(Z_1,\ldots,Z_{m-r}).
\]
\end{theorem}

\begin{proof}
By Theorem~\ref{thm2}, the map
\[
\varepsilon\mapsto y(x,\varepsilon)\]
is analytic near \(0\). Since
\[
I^\varepsilon=I^0+\varepsilon J,
\]
we have
\[
\widehat X_T(x,\varepsilon)=
I^0(y(x,\varepsilon)\widehat Y_T)
+
\varepsilon J(y(x,\varepsilon)\widehat Y_T).
\]

Set
\[
Z_m=y_m\widehat Y_T,
\qquad m\geq0.
\]
Then
\[
y(x,\varepsilon)\widehat Y_T
=
\sum_{m=0}^{\infty}
\frac{\varepsilon^m}{m!}Z_m.
\]


Applying Lemma \ref{lemAPC} with
$$(F_0,F_1)=(I^0,J)$$
and
\[
z(\varepsilon)=y(x,\varepsilon)\widehat Y_T,
\]
gives the coefficient formula.
\end{proof}
 
\subsection{Explicit low-order coefficients: the primal value function}\label{secExpLowOrderP}
The first coefficients of the expansion of the primal value function are
\[
u_0(x)=v_0(y_0)+xy_0,
\]
\[
u_1(x)=v_1(y_0),
\]
and
\[
u_2(x)
=
v_0''(y_0)y_1^2
+
2v_1'(y_0)y_1.
\]
Using
\[
y_1=-\frac{v_1'(y_0)}{v_0''(y_0)},
\]
we obtain
\[
u_2(x)
=
-
\frac{\bigl(v_1'(y_0)\bigr)^2}{v_0''(y_0)}.
\]
In particular, the second-order coefficient is always nonpositive,
since \(v_0\) is strictly convex.

\subsection{Explicit low-order coefficients: the optimal terminal wealth}

Theorem~\ref{thm4} yields
\[
\hat X_T(x,\varepsilon)=
\sum_{m=0}^\infty
\frac{\varepsilon^m}{m!}X_m.
\]

In particular,
\[
X_0=I^0(Z_0),
\]
\[
X_1=(I^0)'(Z_0)Z_1+J(Z_0),
\]
and
\[
X_2=(I^0)''(Z_0)Z_1^2+(I^0)'(Z_0)Z_2+2J'(Z_0)Z_1.
\]


\subsection{First-order perturbations of probabilities and preferences}\label{secJointPert}
Motivated by the discussion in Section \ref{exBS}, we also study two parameter perturbations.
We briefly record a first-order extension in which both the physical
probability measure and the preference specification are perturbed.

Let
\[
\frac{d\mathbb P^\delta}{d\mathbb P}
=
Z_T^\delta,
\qquad
Z^\delta=\mathcal E(\delta M),
\]
where \(Z^\delta\) is a strictly positive true martingale for \(|\delta|\)
sufficiently small. Suppose also that
\[
I^\varepsilon=I^0+\varepsilon J.
\]
For each \((\delta,\varepsilon)\) near \((0,0)\), assume that the market
under \(\mathbb P^\delta\) remains stochastically dominant, and denote by
\(\widehat Y^\delta\) the corresponding distinguished dual element.

For fixed \(x>0\), define \(y_0\) by
\[
\mathbb E\left[
\widehat Y_T^0 I^0(y_0\widehat Y_T^0)
\right]=x.
\]
Assume that, as \(\delta\to0\),
\[
Z_T^\delta=1+\delta Z_1+o(\delta),
\]
and
\begin{equation}\label{defdotY}
\widehat Y_T^\delta
=
\widehat Y_T^0+\delta \dot Y_T+o(\delta),
\end{equation}
for some random variables $Z_1$ and $\dot Y_T$, 
in a topology allowing the differentiations below under the expectation.
Finally assume that
\[
G_y(y_0, 0,0)\neq0,
\]
where
\[
G(y,\delta,\varepsilon)
=
\mathbb E\left[
Z_T^\delta \widehat Y_T^\delta
I^\varepsilon(y\widehat Y_T^\delta)
\right]-x.
\]

\begin{proposition}[First-order expansion of the Lagrange multiplier]\label{prop1}
Under the preceding assumptions, the Lagrange multiplier
\(y(x, \delta,\varepsilon)\), defined by
\[
\mathbb E\left[
Z_T^\delta \widehat Y_T^\delta
I^\varepsilon
\left(y(x,\delta,\varepsilon)\widehat Y_T^\delta\right)
\right]
=
x,
\]
admits the first-order expansion
\[
y(x,\delta,\varepsilon)
=
y_0+\delta y_{1,0}+\varepsilon y_{0,1}
+
o(|\delta|+|\varepsilon|),
\]
where
\[
y_{1,0}
=
-\frac{G_\delta(y_0, 0,0)}{G_y(y_0, 0,0)}
\]
and
\[
y_{0,1}
=
-\frac{G_\varepsilon(y_0, 0,0)}{G_y(y_0, 0,0)}.
\]
Moreover,
\[
G_y (y_0, 0,0) 
=
\mathbb E\left[
(\widehat Y_T^0)^2
(I^0)'(y_0\widehat Y_T^0)
\right]
<0,
\]
\[
G_\varepsilon (y_0, 0,0) 
=
\mathbb E\left[
\widehat Y_T^0J(y_0\widehat Y_T^0)
\right],
\]
and
\[
G_\delta (y_0, 0,0) 
=
\mathbb E\left[
Z_1\widehat Y_T^0I^0(y_0\widehat Y_T^0)
\right]
+
\mathbb E\left[
\dot Y_T
\left(
I^0(y_0\widehat Y_T^0)
+
y_0\widehat Y_T^0(I^0)'(y_0\widehat Y_T^0)
\right)
\right].
\]
Consequently,
\[
y_{0,1}
=
-
\frac{
\mathbb E\left[
\widehat Y_T^0J(y_0\widehat Y_T^0)
\right]
}{
\mathbb E\left[
(\widehat Y_T^0)^2
(I^0)'(y_0\widehat Y_T^0)
\right]
},
\]
and
\[
y_{1,0}
=
-
\frac{
\mathbb E\left[
Z_1\widehat Y_T^0I^0(y_0\widehat Y_T^0)
\right]
+
\mathbb E\left[
\dot Y_T
\left(
I^0(y_0\widehat Y_T^0)
+
y_0\widehat Y_T^0(I^0)'(y_0\widehat Y_T^0)
\right)
\right]
}{
\mathbb E\left[
(\widehat Y_T^0)^2
(I^0)'(y_0\widehat Y_T^0)
\right]
}.
\]
\end{proposition}

\begin{proof}
The Lagrange multiplier is characterized by
\[
G(y (x, \delta,\varepsilon), \delta,\varepsilon)=0.
\]
At \((\delta,\varepsilon)=(0,0)\), this equation becomes
\[
G(y_0, 0,0)
=
\mathbb E\left[
\widehat Y_T^0I^0(y_0\widehat Y_T^0)
\right]-x
=0.
\]
Moreover,
\[
G_y(y_0, 0,0)
=
\mathbb E\left[
(\widehat Y_T^0)^2(I^0)'(y_0\widehat Y_T^0)
\right].
\]
Since \(I^0\in\mathcal C\), we have
\[
(I^0)'(z)<0,\qquad z>0,
\]
and hence
\[
G_y(y_0, 0,0)<0.
\]
Therefore, by the implicit function theorem, \(y(x, \delta,\varepsilon)\)
is differentiable at \((0,0)\), and
\[
y(x, \delta,\varepsilon)
=
y_0+\delta y_{1,0}+\varepsilon y_{0,1}
+
o(|\delta|+|\varepsilon|),
\]
with
\[
y_{1,0}
=
-\frac{G_\delta(y_0, 0,0)}{G_y(y_0, 0,0)},
\qquad
y_{0,1}
=
-\frac{G_\varepsilon(y_0, 0,0)}{G_y(y_0, 0,0)}.
\]

It remains to compute \(G_\delta\) and \(G_\varepsilon\). Since
\[
I^\varepsilon=I^0+\varepsilon J,
\]
we immediately obtain
\[
G_\varepsilon(y_0, 0,0)
=
\mathbb E\left[
\widehat Y_T^0J(y_0\widehat Y_T^0)
\right].
\]
Next, using
\[
Z_T^\delta=1+\delta Z_1+o(\delta),
\qquad
\widehat Y_T^\delta=\widehat Y_T^0+\delta \dot Y_T+o(\delta),
\]
and differentiating
\[
Z_T^\delta \widehat Y_T^\delta
I^0(y_0\widehat Y_T^\delta)
\]
at \(\delta=0\), we obtain
\[
G_\delta(y_0, 0,0)
=
\mathbb E\left[
Z_1\widehat Y_T^0I^0(y_0\widehat Y_T^0)
\right]
+
\mathbb E\left[
\dot Y_T
\left(
I^0(y_0\widehat Y_T^0)
+
y_0\widehat Y_T^0(I^0)'(y_0\widehat Y_T^0)
\right)
\right].
\]
Substituting these expressions into the formulas for \(y_{1,0}\) and
\(y_{0,1}\) completes the proof.
\end{proof}
%
The next proposition gives a first-order expansion of the value function, and is essentially an envelope theorem.
\begin{proposition}[First-order expansion of the value function]\label{prop2}
Assume the conditions of Proposition~\ref{prop1}. Suppose that
\[
v^{\delta,\varepsilon}(y)
=
\E\left[
Z_T^\delta V^\varepsilon(y\widehat Y_T^\delta)
\right],
\]
and that differentiation with respect to \(\delta\) may be passed under the
expectation at \((\delta,\varepsilon)=(0,0)\).

Then
\[
u(x, \delta,\varepsilon)
=
u(x,0,0)
+
\delta u_{1,0}(x)
+
\varepsilon u_{0,1}(x)
+
o(|\delta|+|\varepsilon|),
\]
where
\[
u(x,0,0)
=
v^{0,0}(y_0)+xy_0,
\]
and
\[
u_{0,1}(x)
=
\left.
\frac{\partial}{\partial\varepsilon}
\right|_{\varepsilon=0}
v^{0,\varepsilon}(y_0),
\]
\[
u_{1,0}(x)
=
\left.
\frac{\partial}{\partial\delta}
\right|_{\delta=0}
v^{\delta,0}(y_0).
\]
\end{proposition}
\begin{proof}
By the preceding proposition,
\[
y(x, \delta,\varepsilon)
=
y_0+\delta y_{1,0}+\varepsilon y_{0,1}
+
o(|\delta|+|\varepsilon|).
\]
Using
\[
u(x, \delta,\varepsilon)
=
v^{\delta,\varepsilon}(y(x, \delta,\varepsilon))
+
xy(x, \delta,\varepsilon),
\]
we differentiate at \((0,0)\). The terms involving the derivative of
\(y(x, \delta,\varepsilon)\) cancel because of the first-order condition
\[
v_y^{0,0}(y_0)+x=0.
\]
Therefore the first-order coefficients of the primal value function are
obtained by differentiating only the value function itself with respect to
\(\delta\) and \(\varepsilon\), with \(y=y_0\) fixed.

Hence
\[
u_{0,1}(x)
=
\left.
\frac{\partial}{\partial\varepsilon}
\right|_{\varepsilon=0}
v^{0,\varepsilon}(y_0),
\qquad
u_{1,0}(x)
=
\left.
\frac{\partial}{\partial\delta}
\right|_{\delta=0}
v^{\delta,0}(y_0),
\]
which proves the claim.
\end{proof}
\begin{proposition}[First-order expansion of the optimizer]\label{prop3}
Under the assumptions of Proposition~\ref{prop1}, the family of optimizers 
\[
\widehat X_T(x, \delta,\varepsilon)
=
I^\varepsilon
\bigl(
y(x, \delta,\varepsilon)\widehat Y_T^\delta
\bigr)
\]
admits the first-order expansion
\[
\frac{
\widehat X_T(x, \delta,\varepsilon)
-
\widehat X_T^0
-
\delta X_{1,0}
-
\varepsilon X_{0,1}
}
{|\delta|+|\varepsilon|}
\to0,
\qquad\mathbb P\text{-a.s.},
\]
where 
\[
\widehat X_T^0=I^0(y_0\widehat Y_T^0),
\]
\[
X_{0,1}
=
J(y_0\widehat Y_T^0)
+
y_{0,1}\widehat Y_T^0
(I^0)'(y_0\widehat Y_T^0),
\]
and
\[
X_{1,0}
=
(I^0)'(y_0\widehat Y_T^0)
\left(
y_{1,0}\widehat Y_T^0+y_0\dot Y_T
\right),
\]
and $\dot Y_T$ is given by \eqref{defdotY}.
\end{proposition}
\begin{proof}
In order to expand the optimizer, first, let us observe that 
\[
\widehat X_T(x, \delta,\varepsilon)
=
I^\varepsilon
\left(
y(x, \delta,\varepsilon)\widehat Y_T^\delta
\right).
\]
Next, let us recall that 
\[
I^\varepsilon=I^0+\varepsilon J,
\]
and
\[
y(x, \delta,\varepsilon)\widehat Y_T^\delta
=
y_0\widehat Y_T^0
+
\delta
\left(
y_{1,0}\widehat Y_T^0+y_0\dot Y_T
\right)
+
\varepsilon y_{0,1}\widehat Y_T^0
+
o(|\delta|+|\varepsilon|).
\]
Since \(I^0\) is continuously differentiable and \(J\) is continuous,
\[
I^\varepsilon(z)
=
I^0(z)+\varepsilon J(z),
\]
and a first-order Taylor expansion around
\(z_0=y_0\widehat Y_T^0\) gives
\begin{equation}\nonumber
\hspace{-15mm}
\lim_{|\delta|+|\varepsilon| \to 0}
\frac{
\widehat X_T(x, \delta,\varepsilon)
-
I^0(y_0\widehat Y_T^0)
-
\delta
(I^0)'(y_0\widehat Y_T^0)
\left(
y_{1,0}\widehat Y_T^0+y_0\dot Y_T
\right)
-
\varepsilon
\left(
J(y_0\widehat Y_T^0)
+
y_{0,1}\widehat Y_T^0(I^0)'(y_0\widehat Y_T^0)
\right)
}{
|\delta|+|\varepsilon|
}
=0,\quad a.s.
\end{equation}
which yields the claim.
\end{proof}
\begin{remark}
Unlike the one-parameter setting of Theorems \ref{thm2}--\ref{thm4}, higher-order expansions
for simultaneous perturbations of probabilities and preferences would require
higher-order differentiability, or analyticity, of the map
\[
\delta\mapsto \widehat Y_T^\delta
\]
in a topology compatible with the budget equation, which is a substantially stronger assumption. Consequently, in the present
paper we restrict ourselves to first-order expansions in the two-parameter
framework.
\end{remark}
 
\section{Example: Affine Perturbation of CRRA Inverse Marginal Utilities}\label{secCRRAex}

The perturbations considered in Theorem~1 are formulated directly at the
level of inverse marginal utilities. The purpose of this example is to
illustrate the general theory in the case of power utilities.

Let
\[
p_0<1,
\qquad
p_1<1,
\qquad
p_0\neq 0,
\qquad
p_1\neq 0,
\qquad
p_0\neq p_1.
\]

Consider the inverse marginal utilities
\[
I^{p_0}(y)
=
y^{\frac{1}{p_0-1}},
\qquad
I^{p_1}(y)
=
y^{\frac{1}{p_1-1}},
\qquad y>0.
\]
For $\varepsilon\ge0$, let us set
 
\[
I^\varepsilon
=
I^{p_0}
+
\varepsilon I^{p_1}.
\]

Since both \(I^{p_0}\) and \(I^{p_1}\) are completely monotonic,
\(I^\varepsilon\) is completely monotonic for every \(\varepsilon\ge0\), see also Example \ref{exPower}. The corresponding Bernstein measures belong to $\mathcal M$ by Example \ref{exPower}. In the remainder of this section, the formulas below should be interpreted as right-sided expansions around
\(\varepsilon=0\).

As predicted by Theorem \ref{thm1}, the dual value function admits the affine representation
$$v(y, \varepsilon)=v_0(y)+\varepsilon v_1(y),$$
where $v_
0$
 and $v_1$
 are obtained explicitly from the Bernstein measures. We therefore focus on the corresponding expansions of the Lagrange multiplier, the primal value function, and the optimal terminal wealth.

\subsection{Expansion of the Lagrange multiplier}

The budget equation is
\[
\E\!\left[
\widehat Y_T
I^\varepsilon
\bigl(
y(x, \varepsilon) \widehat Y_T
\bigr)
\right]
=
x.
\]

By Theorem~\ref{thm2},
\[
y(x, \varepsilon)
=
y_0
+
\varepsilon y_1
+
o(\varepsilon),
\]
where \(y_0\) solves
\[
\E\!\left[
\widehat Y_T
I^{p_0}(y_0\widehat Y_T)
\right]
=
x.
\]

Moreover,
\[
y_1
=
-
\frac{
\E\!\left[
\widehat Y_T
I^{p_1}(y_0\widehat Y_T)
\right]
}{
\E\!\left[
\widehat Y^2_T
(I^{p_0})'(y_0\widehat Y_T)
\right]
}.
\]

Since
\[
(I^{p_0})'(z)
=
\frac{1}{p_0-1}
z^{\frac{1}{p_0-1}-1},
\]
and
\[
I^{p_1}(z)
=
z^{\frac{1}{p_1-1}},
\]
we obtain
\[
y_1
=
-
\frac{
\E\!\left[
\widehat Y_T
(y_0\widehat Y_T)^{\frac{1}{p_1-1}}
\right]
}{
\frac{1}{p_0-1}
\E\!\left[
\widehat Y^2_T
(y_0\widehat Y_T)^{\frac{1}{p_0-1}-1}
\right]
}.
\]

Equivalently,
\[
y_1
=
(1-p_0)
\frac{
\E\!\left[
\widehat Y_T
(y_0\widehat Y_T)^{\frac{1}{p_1-1}}
\right]
}{
\E\!\left[
\widehat Y^2_T
(y_0\widehat Y_T)^{\frac{1}{p_0-1}-1}
\right]
}.
\]

\subsection{Expansion of the optimal terminal wealth}

By Theorem~\ref{thm4},
\[
\widehat X_T^\varepsilon
=
\widehat X_T^0
+
\varepsilon X_1
+
o(\varepsilon),
\]
where
\[
\widehat X_T^0
=
I^{p_0}(y_0\widehat Y_T)
=
(y_0\widehat Y_T)^{\frac{1}{p_0-1}},
\]
and
\[
X_1
=
I^{p_1}(y_0\widehat Y_T)
+
y_1\widehat Y_T
(I^{p_0})'(y_0\widehat Y_T).
\]

Substituting the explicit formulas yields
\[
X_1
=
(y_0\widehat Y_T)^{\frac{1}{p_1-1}}
+
\frac{y_1}{p_0-1}
\widehat Y_T
(y_0\widehat Y_T)^{\frac{1}{p_0-1}-1}.
\]


\subsection{Expansion of the primal value function}

By Theorem~\ref{thm3},
\[
u^\varepsilon(x)
=
u^0(x)
+
\varepsilon u_1(x)
+
\frac{\varepsilon^2}{2}u_2(x)
+
o(\varepsilon^2).
\]

The first-order coefficient is
\[
u_1(x)
=
v_1(y_0),
\]
while
\[
u_2(x)
=
-
\frac{
\bigl(
v_1'(y_0)
\bigr)^2
}{
v_0''(y_0)
}.
\]

Therefore, the first two nontrivial coefficients of the expansion are
determined explicitly by the dual quantities \(v_0\) and \(v_1\).



\subsection{Closed-form formulas}

The preceding calculations become considerably more transparent after
introducing the moments of the distinguished dual optimizer
\[
M(\alpha):=\E\!\left[\widehat Y_T^{\,\alpha}\right],\qquad \alpha>0.
\]

 Throughout this subsection we continue to assume
\[
I^{p_0}(y)=y^{\frac1{p_0-1}},
\qquad
J(y)=I^{p_1}(y)=y^{\frac1{p_1-1}},
\]
with
\[
p_0,p_1<1,\qquad p_0,p_1\neq0.
\]

For convenience, let
\[
a:=\frac{p_0}{p_0-1},
\qquad
b:=1+\frac1{p_1-1}.
\]

Then the budget equation becomes
\[
x
=
y_0^{\frac1{p_0-1}}
M(a),
\]
and therefore
\[
y_0
=
\left(
\frac{x}{M(a)}
\right)^{p_0-1}.
\]

Consequently, the optimal terminal wealth of the reference problem is
\[
\widehat X_T
=
(y_0\widehat Y_T)^{\frac1{p_0-1}}
=
\frac{x}{M(a)}
\widehat Y_T^{\frac1{p_0-1}}.
\]

Moreover, the first-order coefficient of the Lagrange multiplier admits the
closed-form representation
\[
y_1
=
(1-p_0)
\,
y_0^{\frac1{p_1-1}-\frac1{p_0-1}+1}
\,
\frac{M(b)}
{M(a)}.
\]

Hence the first-order correction of the optimal terminal wealth is
\[
X_1
=
(y_0\widehat Y_T)^{\frac1{p_1-1}}
+
\frac{y_1}{p_0-1}
\widehat Y_T
(y_0\widehat Y_T)^{\frac1{p_0-1}-1}.
\]

The first-order coefficient of the primal value function is
\[
u_1(x)=v_1(y_0),
\]

while the second-order coefficient simplifies to
\[
u_2(x)
=
-
\frac{\left(v_1'(y_0)\right)^2}
{v_0''(y_0)},
\]
as established in Section~\ref{secExpLowOrderP}.


Thus, the first-order corrections are determined by the moments of the distinguished dual optimizer and the parameters $p_0$ and $p_1$.

\section{A First-Order Extension: Simultaneous Learning of Market Parameters and Preferences}\label{exBS}
In practice, an investor rarely knows either the market model or her own preferences exactly. Market parameters are inferred from historical data and updated over time, while preferences may also evolve as new information becomes available or as the investor's objectives change. This motivates studying the combined effect of learning about the market and changes in investor preferences.  

Motivated by the two-parameter perturbation framework of Section \ref{secJointPert}, we consider a Black–Scholes–Merton model in which both drift estimates and investor preferences vary.
The explicit structure of the model allows the first-order coefficients to be computed in closed form.



Consider a Black--Scholes--Merton market with one riskless asset and one risky
asset. The riskless asset satisfies
\[
dB_t=rB_t\,dt,
\]
where \(r\in\mathbb R\) is constant, while the risky asset satisfies
\[
dS_t=S_t(\mu\,dt+\sigma\,dW_t),
\]
with
\[
\sigma>0.
\]

Suppose that the drift parameter is estimated from data and that, at stage
\(n\), the investor uses an estimate \(\mu_n\), where
\[
\mu_n\to\mu.
\]

At the same time, the investor's preferences are updated and described by
inverse marginal utilities
\begin{equation}\label{671}
I^n(y)
=
I^0(y)+\varepsilon_n I^1(y),
\end{equation}
where
\[
\varepsilon_n\to0.
\]

Thus, two sources of uncertainty are present simultaneously:

\begin{itemize}
\item statistical uncertainty in the estimation of the drift parameter,
represented by
\[
\delta_n:=\mu_n-\mu,
\]

\item uncertainty or learning of preferences, represented by
\[
\varepsilon_n.
\]
\end{itemize}

The corresponding optimal terminal wealth is given by
\[
\widehat X_T^n
=
I^n(y_nH_T^n),
\]
where \(H_T^n\) denotes the state price density corresponding to the model
with drift \(\mu_n\), and \(y_n\) is determined by the budget equation
\[
\E\bigl[H_T^nI^n(y_nH_T^n)\bigr]=x.
\]

%
%
%
%
%

\subsection{State Price Densities}

Define the Sharpe ratios
\[
\theta_n
=
\frac{\mu_n-r}{\sigma},
\qquad
\theta
=
\frac{\mu-r}{\sigma}.
\]

The state price density in model \(n\) is
\[
H_T^n
=
e^{-rT}
\exp \left(
-\theta_nW_T
-\frac12\theta_n^2T
\right).
\]

The limiting state price density is
\[
H_T
=
e^{-rT}
\exp \left(
-\theta W_T
-\frac12\theta^2T
\right).
\]


The corresponding optimal terminal wealth with initial capital \(x>0\) is
given by
\begin{equation}\label{672}
\widehat X_T^n
=
I^n(y_nH_T^n),
\end{equation}
where \(H_T^n\) denotes the state price density corresponding to the model
with drift \(\mu_n\), and \(y_n\) is determined from the budget equation
\[
\mathbb E\left[H_T^nI^n(y_nH_T^n)\right]=x.
\]
Equivalently,
\begin{equation}\label{673}
\mathbb E\left[H_T^n\widehat X_T^n\right]=x,
\end{equation}
so that the sequence of optimization problems corresponds to a fixed
initial wealth \(x\) and varying beliefs and preferences.
In \eqref{672}, the Lagrange multiplier \(y_n\) is determined uniquely by the budget
constraint  \eqref{673}.

\subsection{First-Order Expansion of the State Price Density}

Let
\[
\delta_n=\mu_n-\mu.
\]

Since
\[
\theta_n
=
\theta+\frac{\delta_n}{\sigma},
\]
a Taylor expansion yields
\[
H_T^n
=
H_T
+
\delta_n H_{\mu}
+
o(\delta_n),
\]
where
\[
H_{\mu}
=
\frac{\partial H_T}{\partial \mu}.
\]

Direct differentiation gives
\[
H_{\mu}
=
-\frac{1}{\sigma}
H_T(W_T+\theta T).
\]

\subsection{Preference Perturbations}

In accordance with the framework developed in Sections~1--3, we assume
that the inverse marginal utilities admit the exact affine representation
\[
I^n(y)=I^0(y)+\varepsilon_n I^1(y),
\]
where
\[
\varepsilon_n\to0.
\]

Consequently, the preference perturbation is linear in the parameter
\(\varepsilon_n\), and all nonlinear effects arise through the dependence
of the Lagrange multiplier \(y_n\) on the market and preference
parameters.

We seek an expansion of the form
\[
y_n
=
y+\delta_n y_\mu+\varepsilon_n y_I
+o(|\delta_n|+|\varepsilon_n|),
\]
where \(y\) solves
\[
\E\!\left[H_T I^0(yH_T)\right]=x.
\]

\subsection{Expansion of the Lagrange Multiplier}

Applying the implicit function theorem to the budget equation with respect
to the preference parameter yields
\[
y_I
=
-
\frac{
\E\!\left[H_T I^1(yH_T)\right]
}{
\E\!\left[H_T^2 (I^0)'(yH_T)\right]
}.
\]

Similarly, we get 
\[
y_\mu
=
-
\frac{
\E\!\left[
H_\mu
\Bigl(
I^0(yH_T)+yH_T (I^0)'(yH_T)
\Bigr)
\right]
}{
\E\!\left[
H_T^2 (I^0)'(yH_T)
\right]
}.
\]

\subsection{Expansion of the Optimal Wealth}
The expansion above may be viewed as a two-parameter analogue of
Theorem~\ref{thm4}, with the market perturbation represented by \(\delta_n\)
and the preference perturbation represented by \(\varepsilon_n\).
Denote by
\[
\widehat X_T:=I^0(yH_T)
\]
the optimal terminal wealth corresponding to the limiting model and the
reference preference specification.

Substituting the expansions of \(I^n\), \(H_T^n\), and \(y_n\) into
\[
\widehat X_T^n
=
I^n(y_nH_T^n)
\]
gives

\begin{equation}\label{674}
\widehat X_T^n
=
\widehat X_T
+
\delta_n X_{\mu}
+
\varepsilon_n X_I
+
o(|\delta_n|+|\varepsilon_n|),
\end{equation}
where
\[
\widehat X_T
=
I^0(yH_T),
\]
and
\[
X_I
=I^1(yH_T)
+
y_IH_T(I^0)'(yH_T),
\]
while
\[
X_{\mu}
=
(I^0)'(yH_T)
\Bigl(
y_{\mu}H_T
+
yH_{\mu}
\Bigr).
\]
The coefficients $X_I$
 and $X_\mu$
 represent the sensitivities of the optimal terminal wealth with respect to preference learning and drift learning, respectively.
Equation \eqref{674} separates the first-order effects of preference and market perturbations. We summarize these observations in the following proposition.
 
\begin{proposition}[First-order expansion under simultaneous learning]
Suppose that
\[
H_T^n=H_T+\delta_n H_\mu+o(\delta_n)
\quad\text{a.s.},
\]
and
\[
I^n=I^0+\varepsilon_n I^1,
\qquad
\delta_n\to0,\quad \varepsilon_n\to0.
\]
Assume further that
\[
y_n
=
y+\delta_n y_\mu+\varepsilon_n y_I
+
o(|\delta_n|+|\varepsilon_n|).
\]
Then
\[
\frac{
\widehat X_T^n-\widehat X_T-\delta_nX_\mu-\varepsilon_nX_I
}{
|\delta_n|+|\varepsilon_n|
}
\xrightarrow[n\to\infty]{\text{a.s.}}0.
\]
where
\[
\widehat X_T=I^0(yH_T),
\]
\[
X_I=I^1(yH_T)+y_IH_T(I^0)'(yH_T),
\]
and
\[
X_\mu=(I^0)'(yH_T)(y_\mu H_T+yH_\mu).
\]

\end{proposition}

\begin{proof}
We have
\[
\widehat X_T^n=I^n(y_nH_T^n).
\]
Using the assumed expansions of \(y_n\) and \(H_T^n\),
\[
y_nH_T^n
=
yH_T
+
\delta_n(y_\mu H_T+yH_\mu)
+
\varepsilon_n y_IH_T
+
o(|\delta_n|+|\varepsilon_n|)
\quad\text{a.s.}
\]
Since \(I^0\) is continuously differentiable,
\[
I^0(y_nH_T^n)
=
I^0(yH_T)
+
(I^0)'(yH_T)
\left[
\delta_n(y_\mu H_T+yH_\mu)
+
\varepsilon_n y_IH_T
\right]
+
o(|\delta_n|+|\varepsilon_n|)
\quad\text{a.s.}
\]
Finally, using \(I^n=I^0+\varepsilon_n I_1\), we obtain
\[
\widehat X_T^n
=
I^0(y_nH_T^n)
+
\varepsilon_n I^1(y_nH_T^n).
\]
Since \(y_nH_T^n\to yH_T\) a.s.,
\[
\varepsilon_n I^1(y_nH_T^n)
=
\varepsilon_n I^1(yH_T)
+
o(|\varepsilon_n|)
\quad\text{a.s.}
\]
Combining the preceding identities gives the desired expansion.
\end{proof}

 
Thus, in the complete Black–Scholes–Merton setting, simultaneous perturbations of market parameters and preferences yield explicit first-order sensitivities.

\appendix

\section{Appendix: Bell Polynomials and Fa\`a di Bruno's Formula}\label{sec:Appendix}
We recall the partial Bell polynomials and the Fa\`a di Bruno formula used in Theorems \ref{thm2}--\ref{thm4}.
For integers \(n\geq 0\) and \(0\leq k\leq n\), the partial Bell polynomial
\(B_{n,k}\) is defined by
\[
B_{n,k}(x_1,\ldots,x_{n-k+1})
=
\sum
\frac{n!}
{j_1!\cdots j_{n-k+1}!}
\prod_{m=1}^{n-k+1}
\left(
\frac{x_m}{m!}
\right)^{j_m},
\]
where the summation is over all nonnegative integers
\[
j_1,\ldots,j_{n-k+1}
\]
satisfying
\[
j_1+\cdots+j_{n-k+1}=k,
\]
and
\[
j_1+2j_2+\cdots+(n-k+1)j_{n-k+1}=n.
\]

We use the conventions
\[
B_{0,0}=1,
\qquad
B_{n,0}=0,
\qquad n\geq1.
\]

The first few Bell polynomials are
\[
B_{1,1}(x_1)=x_1,
\]
\[
B_{2,1}(x_1,x_2)=x_2,
\qquad
B_{2,2}(x_1)=x_1^2,
\]
\[
B_{3,1}(x_1,x_2,x_3)=x_3,
\]
\[
B_{3,2}(x_1,x_2)=3x_1x_2,
\]
\[
B_{3,3}(x_1)=x_1^3.
\]

The Bell polynomials arise naturally in the Fa\`a di Bruno formula.
If \(f\) and \(g\) are sufficiently smooth, then
\[
\frac{d^n}{d\varepsilon^n}
f(g(\varepsilon))
=
\sum_{k=1}^{n}
f^{(k)}(g(\varepsilon))
B_{n,k}
\bigl(
g'(\varepsilon),
\ldots,
g^{(n-k+1)}(\varepsilon)
\bigr).
\]

In particular, if
\[
g(\varepsilon)
=
\sum_{m=0}^{\infty}
\frac{\varepsilon^m}{m!}g_m,
\]
then evaluating at \(\varepsilon=0\) yields
\[
\left.
\frac{d^n}{d\varepsilon^n}
f(g(\varepsilon))
\right|_{\varepsilon=0}
=
\sum_{k=1}^{n}
f^{(k)}(g_0)
B_{n,k}
(g_1,\ldots,g_{n-k+1}).
\]

This identity is repeatedly used in the proofs of Theorems~\ref{thm2}--\ref{thm4},
where
\[
g(\varepsilon)=y(x,\varepsilon)\]
and
\[
f=v_0',\quad
f=v_1',\quad
f=v_0,\quad
f=v_1,\quad
f=I,\quad
f=J.
\]
 
\bibliographystyle{alpha}
\bibliography{finance}

\end{document}